\documentclass[format=acmsmall, review=false]{acmart}
\usepackage{acm-ec-19}
\usepackage{booktabs} 
\usepackage[ruled]{algorithm2e} 
\usepackage{tikz}
\usepackage{bm}
\usepackage{setspace}

\DeclareMathOperator*{\argmin}{arg\,min}

\SetAlFnt{\small}
\SetAlCapFnt{\small}
\SetAlCapNameFnt{\small}
\SetAlCapHSkip{0pt}
\IncMargin{-\parindent}

\definecolor{myred}{rgb}{0.8, 0.4, 0.4}
\definecolor{mygreen}{rgb}{0.5, 0.8, 0.65}

\begin{document}
\title{Improved Metric Distortion for Deterministic Social Choice Rules}
\author{Kamesh Munagala}
\affiliation{%
  \institution{Computer Science, Duke University}}
\email{kamesh@cs.duke.edu}

\author{Kangning Wang}
\affiliation{%
  \institution{Computer Science, Duke University}}
\email{knwang@cs.duke.edu}

\begin{abstract}
In this paper, we study the metric distortion of deterministic social choice rules that choose a winning candidate from a set of candidates based on voter preferences. Voters and candidates are located in an underlying metric space. A voter has cost equal to her distance to the winning candidate. Ordinal social choice rules only have access to the ordinal preferences of the voters that are assumed to be consistent with the metric distances. Our goal is to design an ordinal social choice rule with minimum distortion, which is the worst-case ratio, over all consistent metrics, between the social cost of the rule and that of the optimal omniscient rule with knowledge of the underlying metric space. 

The distortion of the best deterministic social choice rule was known to be between $3$ and $5$. It had been conjectured that any rule that only looks at the weighted tournament graph on the candidates cannot have distortion better than $5$. In our paper, we disprove it by presenting a weighted tournament rule with distortion of $4.236$. We design this rule by generalizing the classic notion of uncovered sets, and further show that this class of rules cannot have distortion better than $4.236$. We then propose a new voting rule, via an alternative generalization of uncovered sets. We show that if a candidate satisfying the criterion of this voting rule exists, then choosing such a candidate yields a distortion bound of $3$, matching the lower bound. We present a combinatorial conjecture that implies distortion of $3$, and verify it for small numbers of candidates and voters by computer experiments. Using our framework, we also show that selecting any candidate guarantees distortion of at most $3$ when the weighted tournament graph is cyclically symmetric.

\end{abstract}

\maketitle

\section{Introduction}
Collective decision making and social choice are the cornerstones of modern democracy. Mathematically, a social choice rule is a function that maps ordinal preferences of voters for alternatives (or candidates) to a single winning alternative. One way to characterize social choice rules is by adopting the axiomatic approach. Exemplary works include \cite{gibbard1973manipulation} and \cite{satterthwaite1975strategy}, where several natural desirable properties of social choice rules are proposed, and many of them are shown to be incompatible with each other.

An alternative approach is to adopt the utilitarian view, where agents have latent cardinal preferences underlying the ordinal preferences~\cite{procaccia2006distortion,caragiannis2011voting,boutilier2015optimal}. The goal is to optimize certain social objectives, the most common being the sum of utilities of the voters. The term \emph{distortion} of a social choice rule refers to the worst-case ratio between the objective value that the rule achieves and the optimal one, where the worst case is over all cardinal preferences that are consistent with the given ordinal preferences. However, it has been established that many common rules have unbounded distortion~\cite{procaccia2006distortion} and even randomized social choice rules have distortion of at least $\Omega\left(\sqrt{n}\right)$, where $n$ is the number of candidates~\cite{boutilier2015optimal}.

These impossibility results necessitate imposing additional structure on the underlying cardinal preferences. One common approach is the spatial model, where voters and candidates are located in a certain space (\emph{e.g.} the Euclidean space) and the distance between a voter and a candidate is part of the cost function of the voter~\cite{davis1966some,enelow1984spatial,enelow1990advances,merrill1999unified,schofield2007spatial}. For instance, for the social issue of building a facility serving a neighborhood, it is natural to assume the cost of an individual is the distance between the her and the facility. As another example, consider a government in the process of making a public policy. If people are viewed to be on the left-right political spectrum, the cost of an individual can be modeled as the distance between her and the policy on the spectrum. In both cases, the voters and candidates are located in a metric space and the costs are the corresponding distances.

When the cardinal costs are assumed to be distances in an underlying metric space, it is possible to achieve constant distortion. The work of~\cite{anshelevich2018approximating} shows that selecting any candidate in the uncovered set~\cite{moulin1986choosing} --- for instance, the Copeland voting rule --- guarantees distortion of $5$. (Please refer to Section~\ref{sec:rules} for description of several common voting rules.) They also show a lower bound of $3$ for the distortion for any (deterministic) social choice rule. The exact value of the optimal distortion remains an open question --- neither the upper bound of $5$ nor the lower bound of $3$ has been improved. The work of~\cite{skowron2017social} shows that the well known voting schemes of using scoring rules, as well as the popular single transferable vote, have super-constant distortion. The work of~\cite{goel2017metric} shows the ranked pairs rule and the Schulze method, both of which work on the weighted tournament graph over candidates, have distortion of at least $5$ contrary to previous belief that they could improve the bound of $5$. The weighted tournament graph places a directed edge between two candidates with weight equal to the portion of voters who prefer the first candidate over the second. The work of~\cite{goel2017metric} then conjectured any social choice rule that only looks at the weighted tournament graph over candidates cannot beat $5$.

\subsection{Our Results}
Our first result in Section~\ref{sec:4.236} disproves the conjecture that weighted tournament rules cannot improve the distortion bound below $5$. We improve the distortion upper bound to $2 + \sqrt{5} \approx 4.236$ using a weighted tournament rule. We achieve the bound by generalizing the notion of uncovered sets to a class of weighted rules. We show $4.236$ is indeed the best distortion achievable by this new class of social choice rules. 

This still leaves a gap between the best lower bound of $3$ and the upper bound of $4.236$ for deterministic social choice rules.  In Section~\ref{sec:matching}, we propose another social choice rule based on bipartite matchings that generalizes uncovered sets in a different way. In Theorem~\ref{theorem:uncovered_set_3}, we show that if a candidate exists that satisfies our criterion, then choosing this candidate yields distortion of $3$. In Section~\ref{sec:conjecture}. we present a combinatorial conjecture that implies such a candidate always exists, and verify this conjecture for small numbers of voters and candidates (at most $14$ in all). As another advantage of using this framework, we show a new result in Section~\ref{sec:symmetric} for certain types of preference profiles: If the weighted tournament graph on the instance is cyclically symmetric, then selecting any candidate yields distortion of at most $3$ for that profile.

\subsection{Related Work}
\paragraph{Distortion of randomized social choice rules}
In addition to deterministic social choices rules, randomized rules have been studied in the metric distortion setting. It is possible that they can achieve better distortion than deterministic ones: Randomly selecting a voter to be the dictator achieves distortion of $3$~\cite{feldman2016voting,anshelevich2017randomized}. On the other hand, no randomized rule can get distortion better than $2$~\cite{feldman2016voting,anshelevich2017randomized}, no truthful randomized rule can beat $3$~\cite{feldman2016voting} and no weighted tournament rule can beat $3$~\cite{goel2017metric}. An almost distortion-optimal rule when the number of candidates is not large is shown in~\cite{gross17vote}. Randomized rules are unsatisfying for practical implementation and interpretability, and much literature on social choice has therefore focused on deterministic voting rules.

\paragraph{Median distortion and fairness properties}
In our paper, we define the social cost to be the sum of the cost of each individual voter. There are other measures of how well a social choice rule performs. The median distortion setting, where distortion is defined using the median cost instead of the sum is considered in~\cite{anshelevich2018approximating}. They show the Copeland rule achieves distortion of $5$ in the setting and is optimal. Fairness properties of social choice rules generalizes distortion in the sense that for any $k$, it considers the sum of the $k$ largest costs~\cite{goel2017metric}. For this setting, the Copeland rule achieves a fairness ratio of $5$. Further, fairness ratio and distortion can only differ by at most $2$ additively~\cite{goel2018relating}. We show in Section~\ref{sec:lower_bound} that though our voting rule improves distortion for the sum of cost, it does not improve distortion bounds for fairness properties.

\paragraph{Other Problems in the Metric Distortion Setting}
The metric distortion framework is powerful enough to be applied to settings beyond that considered in this work. When candidates are independently drawn from the population of voters, the distortion can be better than $2$~\cite{cheng2017people} and the class of scoring rules that have constant expected distortion has a clean characterization~\cite{cheng2018distortion}.
Low-distortion algorithms for the maximum bipartite matching problem in metric spaces are proposed in \cite{anshelevich2017tradeoffs}. Social choice and facility location problems in the distortion framework are studied in~\cite{anshelevich2018ordinal}.

\section{Preliminaries}
\subsection{Social Choice Rules}
Let $\mathcal{C}= \{1, 2, \ldots, n\}$ be a set of $n$ candidates and $\mathcal{V} = \{1', 2', \ldots, m'\}$ be a set of $m$ voters. Let $\mathcal{A} = \mathcal{C} \cup \mathcal{V}$ be the set of entities in the voting system. For each voter $v \in \mathcal{V}$, she has a strict preference ordering $\sigma_v$ on the candidates in $\mathcal{C}$. We call $\bm{\sigma} = \{\sigma_v\}_{v \in \mathcal{V}}$ the voting profile. If $v \in \mathcal{V}$ prefers $c_1 \in \mathcal{C}$ to $c_2 \in \mathcal{C}$, we write $c_1 \succ_v c_2$ or $c_2 \prec_v c_1$. Let $\mathcal{O}_k$ be the family of strict orderings on $k$ elements. A \emph{social choice rule} is a function
$$f : \bigcup_{n, m = 1}^\infty \left(\mathcal{O}_n\right)^m \to \mathbb{N}$$
that takes $m$ orderings on $n$ candidates and outputs one of the $n$ candidates. The output of $f$ is the winner of the social choice rule.

Throughout the paper, we reserve upper-case letters $A, B, C, D, X, Y, Z$ to denote candidates. $AB$ is the set of voters that prefer $A$ to $B$. $ABC$ is the set of voters that both prefer $A$ to $B$ and prefer $B$ to $C$. For any set $S$, we use $|S|$ to denote the size of $S$. \emph{E.g.}, $|AX|$ is the number of voters who prefer $A$ to $X$.

We will need the following definition later.
\begin{definition}
For any voter $v \in \mathcal{V}$ and any candidate $X \in \mathcal{C}$, $P_v(X) := \{X\} \cup \{Y \in \mathcal{C} \mid Y \succ_v X\}$ is the set of candidates that $v$ likes at least as much as $X$. Similarly, $Q_v(X) := \{X\} \cup \{Y \in \mathcal{C} \mid Y \prec_v X\}$ is the set of candidates that $v$ likes at most as much as $X$.
\label{def:PQ}
\end{definition}

\subsection{Metric Distortion}
We assume the candidates and the voters are located in the same metric space $(\mathcal{A}, d)$, where $d : \mathcal{A} \times \mathcal{A} \to \mathbb{R}_{\geq 0}$ is the metric (distance function). As a metric, $d$ satisfies the following conditions:
\begin{enumerate}
\item Positive definiteness: $\forall i, j \in \mathcal{A}, d(i, j) \geq 0 \text{ and } d(i, j) = 0 \iff i = j$.
\item Symmetry: $\forall i, j \in \mathcal{A}, d(i, j) = d(j, i)$.
\item Triangle inequality: $\forall i, j, k \in \mathcal{A}, d(i, j) + d(j, k) \geq d(i, k)$.
\end{enumerate}

We say a metric $d$ is \emph{consistent} with a voting profile $\bm{\sigma}$ if
\[
\forall v \in \mathcal{V}, \forall A, B \in \mathcal{C}: A \succ_v B \implies d(A, v) \leq d(B, v).
\]
Here we allow a voter to have equal distances to different candidates. This is only to eliminate the need of infinitesimal distances in later analysis. Similar results can be derived if we require $d(A, v)$ and $d(B, v)$ to be different in the above definition.

We define the \emph{social cost} of a candidate $X$ with respective to a metric $d$ as:
\[
S(X, d) = \sum_{v \in \mathcal{V}} d(v, X)
\]
and the \emph{distortion} of a social choice rule $f$ as:
\[
\Delta(f) = \sup_{\bm{\sigma}} \sup_{d \text{ consistent with } {\bm{\sigma}}} \frac{S(f(\bm{\sigma}), d)}{\min_{A \in \mathcal{C}}S(A, d)}.
\]

\subsection{Weighted Tournament Graphs}
We can construct a tournament graph on the candidates by looking at their pairwise relationships: Let $\mathcal{C}$ be the set of vertices. For each pair of $A, B \in \mathcal{C}$ ($A \neq B$), draw an edge from $A$ to $B$ if $A$ pairwise beats $B$: $|AB| = |\{v \in \mathcal{V} : A \succ_v B\}| \geq \frac{m}{2}$ (break ties arbitrarily). The result is a directed graph where there is exactly one directed edge between any pair of vertices. If a social choice rule makes decision only based on the tournament graph, it is called a \emph{tournament rule}, or a $C1$ rule~\cite{fishburn1977condorcet}.

One way to generalize the class of tournament rules is to further consider the margin of each winning. For each pair of $A, B \in \mathcal{C}$ ($A \neq B$), draw an edge from $A$ to $B$ with weight $\frac{|AB|}{m}$. In the resulted \emph{weighted tournament graph}, there are two directed edges between any pair of vertices. A social choice rule making decision only based on the weighted tournament graph is called a \emph{weighted tournament rule}, or a $C2$ rule\footnote{To be precise, in \cite{fishburn1977condorcet}, a social choice rule is $C2$ only if it is not $C1$.}~\cite{fishburn1977condorcet}. We also consider a subclass of weighted tournament graphs that satisfy \emph{cyclic symmetry}.

\begin{definition}
We say a graph is \emph{cyclically symmetric}, if there exists a cyclic permutation $\tau$ on the vertices, so that any edge $(u, v)$ has the same weight as $(\tau(u), \tau(v))$.
\end{definition}
\begin{example}
Let $\mathcal{C} = \{A, B, C, D, E\}$. Figure~\ref{fig:symmetric} shows a cyclically symmetric weighted tournament graph. For each pair of vertices, only one edge is shown for visibility. The cyclic permutation $\tau$ is
\begin{equation*} 
 \tau = \begin{pmatrix}
  A & B & C & D & E \\
  B & C & D & E & A
 \end{pmatrix}.
\end{equation*}
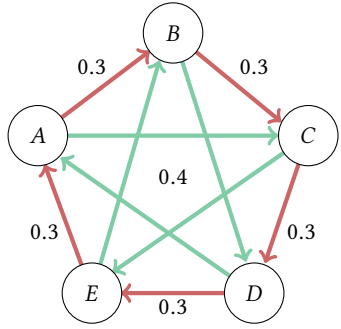
\begin{figure}
\centering
\begin{tikzpicture}[scale=1.8]
\draw [myred,ultra thick,->] (0, 0) to (0.83, 0.62);
\draw [myred,ultra thick,->] (1, 0.75) to (1.8, 0.1);
\draw [myred,ultra thick,->] (2, 0) to (1.68, -0.93);
\draw [myred,ultra thick,->] (1.6, -1.15) to (0.61, -1.15);
\draw [myred,ultra thick,->] (0.4, -1.15) to (0.05, -0.21);

\draw [mygreen,ultra thick,->] (0, 0) to (1.78, 0);
\draw [mygreen,ultra thick,->] (1, 0.75) to (1.54, -0.92);
\draw [mygreen,ultra thick,->] (2, 0) to (0.55, -1.00);
\draw [mygreen,ultra thick,->] (1.6, -1.15) to (0.17, -0.16);
\draw [mygreen,ultra thick,->] (0.4, -1.15) to (0.9, 0.55);

\draw [fill=white](0,0) circle [radius=0.22];
\draw [fill=white](1,0.75) circle [radius=0.22];
\draw [fill=white](2,0) circle [radius=0.22];
\draw [fill=white](1.6,-1.15) circle [radius=0.22];
\draw [fill=white](0.4,-1.15) circle [radius=0.22];

\node at (0,0) {\small $A$};
\node at (1,0.75) {\small $B$};
\node at (2,0) {\small $C$};
\node at (1.6,-1.15) {\small $D$};
\node at (0.4,-1.15) {\small $E$};

\node at (0.4,0.5) {\small $0.3$};
\node at (1.6,0.5) {\small $0.3$};
\node at (1.95,-0.7) {\small $0.3$};
\node at (0.05,-0.7) {\small $0.3$};
\node at (1,-1.25) {\small $0.3$};

\node at (1,-0.3) {\small $0.4$};

\end{tikzpicture}
\caption{A Cyclically Symmetric Weighted Tournament Graph}
\label{fig:symmetric}
\end{figure}
\end{example}
We note that cyclically symmetric weighted tournament graph is not necessarily induced by a cyclically symmetric voting profile.

\subsection{Examples of Social Choice Rules}
\label{sec:rules}
Now we present some social choice rules from existing literature, along with their distortion bounds.
\paragraph{\textsc{Uncovered}} The uncovered set of a tournament graph consists of those candidates $A \in \mathcal{C}$ such that for any other candidate $B$, either there is an edge from $A$ to $B$, or there is another candidate $C$ so that there are an edge from $A$ to $C$ and an edge from $C$ to $B$~\cite{moulin1986choosing}. The uncovered set is always nonempty for any tournament graph. \textsc{Uncovered} is a $C1$ rule that selects the alphabetically smallest candidate from the uncovered set. It has distortion of $5$~\cite{anshelevich2018approximating}.
\paragraph{\textsc{Copeland}} \textsc{Copeland} selects the candidate that has the most pairwise wins, \emph{i.e.}, the candidate with the largest out-degree in the tournament graph (break ties alphabetically). \textsc{Copeland} is a $C1$ rule as well, and it can be seen as \textsc{Uncovered} with a specific tie-breaking method. It does not improve the $5$-distortion of \textsc{Uncovered}~\cite{anshelevich2018approximating}.
\paragraph{\textsc{RankedPairs}} \textsc{RankedPairs} is a $C2$ rule proposed in \cite{tideman1987independence}. In \textsc{RankedPairs}, the edges of the weighted tournament graph are sorted to have decreasing weights. We start from an empty graph and add one edge at a time (in the order of decreasing weights) if it does not create a cycle. A directed acyclic graph (DAG) is resulted in the end, and the source of the DAG is selected as the winning candidate. It achieves distortion of $3$ when the tournament graph has a circumference (size of the largest cycle) of at most $4$~\cite{anshelevich2018approximating}. However, its distortion is at least $5$ in general~\cite{goel2017metric}.
\paragraph{\textsc{Schulze}} \textsc{Schulze} is a $C2$ rule proposed in \cite{schulze2011new}. For two candidates $A$ and $B$ in the weighted tournament graph, define $p(A, B)$ to be the maximum $p$ so that there is a path from $A$ to $B$ with each edge having weight of at least $p$. It can be shown that there is a candidate $X \in \mathcal{C}$ so that $p(X, Y) \geq p(Y, X)$ for any other $Y \in \mathcal{C}$. This $X$ is selected as the winning candidate by \textsc{Schulze}. Its distortion is at least $5$~\cite{goel2017metric}.
\paragraph{\textsc{OptimalLP}} In \cite{goel2017metric}, the authors show that the optimal distortion can be achieved by \textsc{OptimalLP} involving solving linear programs. The distortion of selecting candidate $A$ comparing to selecting $B$ in an instance is given by the following linear program:
\begin{equation*}
\onehalfspacing
\begin{array}{lll@{}ll}
P(A, B, \bm{\sigma})  := &\text{maximize}   &\displaystyle \sum_{v \in \mathcal{V}} d(A, v) \\[\bigskipamount]
&\text{subject to}& \displaystyle \sum_{v \in \mathcal{V}}   d(B, v) = 1,\\
                 &                               &d \text{ is consistent with } \bm{\sigma}.
\end{array}
\end{equation*}
\textsc{OptimalLP} selects the candidate $\argmin_{A \in \mathcal{C}} \max_{B \in \mathcal{C}} P(A, B, \bm{\sigma})$. Its distortion was known to be between $3$ and $5$ (inclusive). Our work implies \textsc{OptimalLP} has distortion of at most $2 + \sqrt{5}$.

\section{Weighted-Uncovered-Set Voting Rule}
\label{sec:4.236}
In this section, we present the \textsc{WeightedUncovered} rule that has distortion of $2 + \sqrt{5} \approx 4.236$ (Theorem~\ref{theorem:4.236}), improving the best existing upper bound of $5$ for any deterministic social choice rule. \textsc{WeightedUncovered} is based on the \emph{$\lambda$-weighted uncovered set} (Definition~\ref{def:WUS}), a generalization of the uncovered set~\cite{moulin1986choosing}. We show our analysis is tight by giving a matching lower bound (Theorem~\ref{theorem:lower}). We also give a lower bound for the fairness ratio (Theorem~\ref{theorem:fairness_lower}).

We start by introducing \emph{$\lambda$-weighted uncovered sets}.
\begin{definition}
Let $\lambda \in [0.5, 1]$ be a constant. The \emph{$\lambda$-weighted uncovered set} is the set of candidates $A \in \mathcal{C}$ such that for any candidate $B \neq A$, we have:
\begin{itemize}
\item either $|AB| \geq (1 - \lambda) m$,
\item or there is another node $C \notin \{A, B\}$ so that $|AC| \geq (1 - \lambda) m$ and $|CB| \geq \lambda m$.
\end{itemize}
\label{def:WUS}
\end{definition}

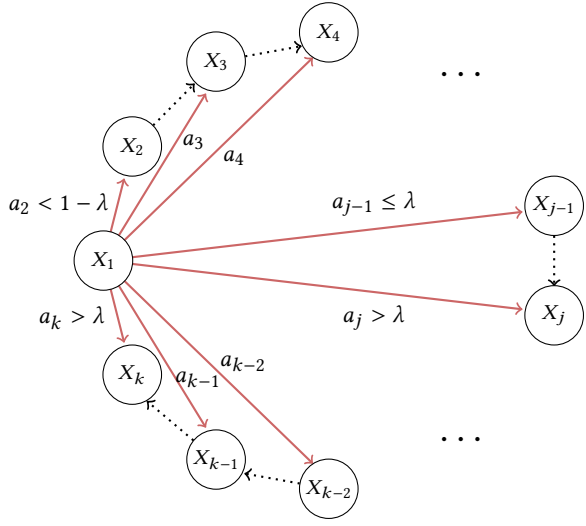
\begin{figure}[H]
\centering
\begin{tikzpicture}[scale=0.3]
\draw [thick,myred,->] (0, 0) to (1.28*0.72,5*0.72);
\draw [thick,myred,->] (0, 0) to (5*0.91,8.72*0.84);
\draw [thick,myred,->] (0, 0) to (10*0.93,10*0.89);
\draw [thick,myred,->] (0, 0) to (20*0.93,2.4*0.89);
\draw [thick,myred,->] (0, 0) to (1.28*0.72,-5*0.72);
\draw [thick,myred,->] (0, 0) to (5*0.91,-8.72*0.84);
\draw [thick,myred,->] (0, 0) to (10*0.93,-10*0.89);
\draw [thick,myred,->] (0, 0) to (20*0.93,-2.4*0.89);

\draw [thick,dotted,->] (1.28,5) to (5*0.81,8.72*0.89);
\draw [thick,dotted,->] (5,8.72) to (10*0.88,10*0.94);
\draw [thick,dotted,->] (5,-8.72) to (1.28*1.5,-5*1.23);
\draw [thick,dotted,->] (10,-10) to (5*1.25,-8.72*1.07);
\draw [thick,dotted,->] (20,2.4) to (20*1,-2.4*0.47);

\draw [fill=white](0,0) circle [radius=1.3];
\draw [fill=white](1.28,5) circle [radius=1.3];
\draw [fill=white](5,8.72) circle [radius=1.3];
\draw [fill=white](10,10) circle [radius=1.3];
\draw [fill=white](1.28,-5) circle [radius=1.3];
\draw [fill=white](5,-8.72) circle [radius=1.3];
\draw [fill=white](10,-10) circle [radius=1.3];
\draw [fill=white](20,2.4) circle [radius=1.3];
\draw [fill=white](20,-2.4) circle [radius=1.3];

\node at (0,0) {\footnotesize $X_1$};
\node at (1.28,5) {\footnotesize $X_2$};
\node at (5,8.72) {\footnotesize $X_3$};
\node at (10,10) {\footnotesize $X_4$};
\node at (1.28,-5) {\footnotesize $X_{k}$};
\node at (5,-8.72) {\footnotesize $X_{k-1}$};
\node at (10,-10) {\footnotesize $X_{k-2}$};
\node at (20,2.4) {\footnotesize $X_{j-1}$};
\node at (20,-2.4) {\footnotesize $X_{j}$};

\node at (16,8) {\LARGE $\cdots$};
\node at (16,-8) {\LARGE $\cdots$};

\node at (-2,2.5) {\small $a_2 < 1 - \lambda$};
\node at (-1.4,-2.5) {\small $a_k > \lambda$};

\node at (4,5.2) {\small $a_3$};
\node at (4.2,-5.2) {\small $a_{k-1}$};
\node at (5.8,4.5) {\small $a_4$};
\node at (6.2,-4.5) {\small $a_{k-2}$};

\node at (12,2.5) {\small $a_{j-1} \leq \lambda$};
\node at (12,-2.5) {\small $a_{j} > \lambda$};
\end{tikzpicture}
\caption{The Cycle in Lemma~\ref{lem:weighted_uncovered_nonempty}}
\label{fig:weighted_uncovered_nonempty}
\end{figure}

Now we present our key lemma showing the $\lambda$-weighted uncovered set is a meaningful concept.
\begin{lemma}
For any $\lambda \in [0.5, 1]$, the $\lambda$-weighted uncovered set is nonempty.
\label{lem:weighted_uncovered_nonempty}
\end{lemma}
\begin{proof}
We prove by contradiction. Suppose there is no such node $A$. Then there must be nodes $X_1, X_2, \ldots, X_k, X_{k+1} = X_1$ in the graph. If we set $A = X_i$ and $B = X_{i+1}$ ($i = 1, 2, \ldots, k$), then both of the conditions in the lemma are violated. Figure~\ref{fig:weighted_uncovered_nonempty} shows such a cycle.

Let $a_j = \frac{|X_1X_j|}{m}$. We have $a_k = \frac{|X_1X_k|}{m} > \lambda$ and $a_2 = \frac{|X_1X_2|}{m} < 1 - \lambda$. When $j$ goes from $k$ down to $2$, $a_j$ has to drop below $\lambda$ at some time, so we have $a_j > \lambda$ while $a_{j-1} \leq \lambda$. Thus we have $|X_1X_j| > \lambda m$ and $|X_{j-1}X_1| \geq (1 - \lambda) m$. This contradicts with the assumption that when $A = X_{j-1}$ and $B = X_j$, there is no such $C = X_1$ satisfying the second condition in the lemma.
\end{proof}

Throughout this section, we use $\varphi$ to denote the golden ratio $\frac{\sqrt{5} - 1}{2} \approx 0.618$.
\begin{definition}
\textsc{WeightedUncovered} is a social choice rule that picks the alphabetically smallest candidate in the $\varphi$-weighted uncovered set.
\end{definition}

Because of Lemma~\ref{lem:weighted_uncovered_nonempty}, \textsc{WeightedUncovered} is well-defined and always picks a candidate in the $\varphi$-weighted uncovered set. We state the main property of \textsc{WeightedUncovered} below in Theorem~\ref{theorem:4.236} and prove it later.
\begin{theorem}
\textsc{WeightedUncovered} has distortion of at most $\left(2 + \sqrt{5}\right) \approx 4.236$.
\label{theorem:4.236}
\end{theorem}

\subsection{Main Theorem: Improved Distortion Bound}
Recall that $S(X, d) = \sum_{v \in \mathcal{V}} d(v, X)$ is the social cost if $X$ is selected and the metric is $d$. First, if $|AB| \geq (1 - \varphi) m$, the social cost gap between choosing $A$ and $B$ can be easily bounded.
\begin{lemma}
If $|AB| \geq (1 - \varphi) m$, then $S(A, d) \leq \left(2 + \sqrt{5}\right)S(B, d)$ for any metric $d$.
\label{lem:one_edge}
\end{lemma}
\begin{proof}
We have
\begin{align*}
\left(2 + \sqrt{5}\right) S(B, d) - S(A, d) &= \sum_{v \in AB} \left(\left(2 + \sqrt{5}\right) d(B, v) - d(A, v)\right) + \sum_{v \in BA} \left(\left(2 + \sqrt{5}\right) d(B, v) - d(A, v)\right)\\
&\geq \sum_{v \in AB} \left(1 + \sqrt{5}\right) d(B, v) - \sum_{v \in BA} d(A, B)\\
&\geq \sum_{v \in AB} \frac{1 + \sqrt{5}}{2} d(A, B) - \sum_{v \in BA} d(A, B)\\
&= d(A, B) \cdot \left(\frac{3 + \sqrt{5}}{2} |AB| - m\right)\\
&\geq 0. \qedhere
\end{align*}
\end{proof}

When showing \textsc{Copeland} has distortion of at most $5$, \cite{anshelevich2018approximating} crucially uses the following lemma:
\begin{lemma}[Theorem 15 in~\cite{anshelevich2018approximating}]
If $|AC| \geq \frac{1}{2} \cdot m$ and $|CB| \geq \frac{1}{2} \cdot m$, then $S(A, d) \leq 5S(B, d)$ for any metric $d$.
\label{lem:half_half}
\end{lemma}
Lemma~\ref{lem:half_half} itself is tight, but worst-case analysis is deployed when we assume every edge in the unweighted tournament graph has weight of $0.5$. Actually, if every edge has weight approaching $0.5$, any social choice rule will have distortion approaching $3$ by a similar argument as Lemma~\ref{lem:one_edge}. Instead of having the stronger condition in the uncovered set that $|AB| \geq 0.5m$, we allow those candidate $A$'s with $|AB| \geq (1 - \lambda)m$ to be in the $\lambda$-weighted uncovered set. In this way, we can use unbalanced $|AC|$ and $|CB|$ and still guarantee the $\lambda$-weighted uncovered set is nonempty (Lemma~\ref{lem:weighted_uncovered_nonempty}). It turns out that better distortion can be achieved with carefully chosen parameters.
\begin{lemma}
If $|AC| \geq (1 - \varphi) m$ and $|CB| \geq \varphi m$, then $S(A, d) \leq \left(2 + \sqrt{5}\right)S(B, d)$ for any metric $d$.
\label{lem:two_edges}
\end{lemma}
To prove Lemma~\ref{lem:two_edges}. We need the following technical lemma from~\cite{anshelevich2018approximating}, whose proof is included here for completeness. Recall that $Q_v(A)$ is the set of candidates that $v$ likes at most as much as $A$.
\begin{lemma}[Lemma 14 in~\cite{anshelevich2018approximating}]
If $\sum_{v \in \mathcal{V}} d(v, B) \geq \frac{1}{\gamma} \sum_{v \in BA} \min_{C \in Q_v(A)} d(B, C)$, then $S(A, d) \leq (1 + \gamma)S(B, d)$ for any metric $d$.
\label{lem:lem_from_1}
\end{lemma}
\begin{proof}
By definition of $Q_v(A)$ and the triangle inequality,
\begin{align*}
\frac{S(A, d)}{S(B, d)} &= \frac{\sum_{v \in \mathcal{V}} d(A, v)}{\sum_{v \in \mathcal{V}} d(B, v)} = \frac{\sum_{v \in AB} d(A, v) + \sum_{v \in BA} d(A, v)}{\sum_{v \in \mathcal{V}} d(B, v)}\\
&\leq \frac{\sum_{v \in AB} d(B, v) + \sum_{v \in BA} (d(A, v) - d(B, v) + d(B, v))}{\sum_{v \in \mathcal{V}} d(B, v)}\\
&\leq \frac{\sum_{v \in \mathcal{V}} d(B, v) + \sum_{v \in BA} \min_{C \in Q_v(A)} (d(C, v) - d(B, v))}{\sum_{v \in \mathcal{V}} d(B, v)}\\
&\leq 1 + \frac{\sum_{v \in BA} \min_{C \in Q_v(A)} d(B, C)}{\sum_{v \in \mathcal{V}} d(B, v)} \leq 1 + \gamma.\qedhere
\end{align*}
\end{proof}

\begin{proof}[Proof of Lemma~\ref{lem:two_edges}]
We perform different analysis depending on whether $d(C, B) \geq d(A, B)$ or $d(C, B) < d(A, B)$.

\noindent\textbf{Case (1):} When $d(C, B) \geq d(A, B)$:
\begin{align*}
\sum_{v \in \mathcal{V}} d(v, B) &\geq \sum_{v \in CB} d(v, B) \geq \frac{1}{2}\sum_{v \in CB} d(C, B) = \frac{1}{2} |CB| \cdot d(C, B).
\end{align*}
Using the assumption that $d(C, B) \geq d(A, B)$ in this case,
\begin{align*}
\sum_{v \in \mathcal{V}} d(v, B) &\geq \frac{\varphi m}{2} \cdot d(A, B) \geq \frac{\varphi}{2} \cdot \sum_{v \in BA} d(A, B).
\end{align*}
Then we apply Lemma~\ref{lem:lem_from_1} to get $S(A, d) \leq \left(1 + \frac{2}{\varphi}\right) S(B, d) = \left(2 + \sqrt{5}\right) S(B, d)$.

\noindent\textbf{Case (2):} When $d(C, B) < d(A, B)$:
\begin{align*}
\sum_{v \in \mathcal{V}} d(v, B) &\geq \sum_{v \in AB \cup BAC \cup CBA} d(v, B)\\
&\geq \frac{1}{2} \sum_{v \in AB} d(A, B) + \sum_{v \in BAC} d(v, B) + \frac{1}{2} \sum_{v \in CBA} d(B, C).
\end{align*}
Notice that for $v \in BAC$, $d(v, B) + d(B, C) \geq d(v, C) \geq d(v, A) \geq d(A, B) - d(v, B)$. Thus,
\begin{align*}
\sum_{v \in \mathcal{V}} d(v, B) &\geq \frac{1}{2} \sum_{v \in AB} d(A, B) + \frac{1}{2} \sum_{v \in BAC} (d(A, B) - d(B, C)) + \frac{1}{2} \sum_{v \in CBA} d(B, C)\\
&= \frac{1}{2} \left((|AB| + |BAC|) \cdot d(A, B) + (|CBA| - |BAC|) \cdot d(B, C)\right)\\
&= \frac{1}{2} \left((|AB| + |BAC|) \cdot (d(A, B) - d(B, C)) + (|CBA| + |AB|) \cdot d(B, C)\right).
\end{align*}
Because $AC \subseteq BAC \cup AB$ and $CB \subseteq CBA \cup AB$,
\begin{align*}
\sum_{v \in \mathcal{V}} d(v, B) &\geq \frac{1}{2} \left(|AC| \cdot (d(A, B) - d(B, C)) + |CB| \cdot d(B, C)\right)\\
&\geq \frac{1}{2} \left((1 - \varphi) m \cdot (d(A, B) - d(B, C)) + \varphi m \cdot d(B, C)\right).
\end{align*}
The second inequality above comes from the conditions that $|AC| \geq (1 - \varphi)m$, $|CB| \geq \varphi m$ and the assumption that $d(A, B) > d(B, C)$ in this case. Therefore,
\begin{align*}
\sum_{v \in \mathcal{V}} d(v, B) &\geq \frac{1}{2} \left(\frac{1 - \varphi}{\varphi} \cdot |CA| \cdot (d(A, B) - d(B, C)) + \varphi \cdot |BA| \cdot d(B, C)\right)\\
&\geq \frac{\varphi}{2} \left((|CBA| + |BCA|) \cdot (d(A, B) - d(B, C)) + (|CBA| + |BCA| + |BAC|) \cdot d(B, C)\right)\\
&= \frac{\varphi}{2} \left((|CBA| + |BCA|) \cdot d(A, B) + |BAC| \cdot d(B, C)\right)\\
&= \frac{\varphi}{2} \left(\sum_{v \in CBA \cup BCA} d(A, B) + \sum_{v \in BAC} d(B, C)\right),
\end{align*}
where the second inequality is from $CBA \cup BCA \subseteq CA$ and $CBA \cup BCA \cup BAC = BA$. Again by $BA = CBA \cup BCA \cup BAC$, we apply Lemma~\ref{lem:lem_from_1} to get $S(A, d) \leq \left(1 + \frac{2}{\varphi}\right) S(B, d) = \left(2 + \sqrt{5}\right) S(B, d)$.
\end{proof}

\begin{proof}[Proof of Theorem~\ref{theorem:4.236}]
By Lemma~\ref{lem:weighted_uncovered_nonempty}, \textsc{WeightedUncovered} always selects a candidate in the uncovered set. By Lemma~\ref{lem:one_edge}, Lemma~\ref{lem:two_edges} and the definition of \textsc{WeightedUncovered}, \textsc{WeightedUncovered} has distortion of at most $2 + \sqrt{5}$.
\end{proof}

\subsection{Lower Bounds}
\label{sec:lower_bound}
An immediate question is whether our upper bound is tight. Also, is it possible to choose a $\lambda \neq \varphi$ in the $\lambda$-weighted uncovered set to improve the distortion? In this section, we show that our choice of $\lambda = \varphi$ and our analysis are indeed tight. Before showing the lower bounds, we complete the definition of $\lambda$-weighted uncovered sets for $\lambda < 0.5$.
\begin{definition}
Let $\lambda \in [0, 0.5)$ be a constant. The \emph{$\lambda$-weighted uncovered set} is the set of candidates $A \in \mathcal{C}$ such that for any candidate $B \neq A$, we have:
\begin{itemize}
\item either $|AB| \geq \lambda m$ (\emph{different from the case where $\lambda \geq 0.5$}),
\item or there is another node $C \notin \{A, B\}$ so that $|AC| \geq (1 - \lambda) m$ and $|CB| \geq \lambda m$.
\end{itemize}
\end{definition}
We did not define $\lambda$-weighted uncovered sets for $\lambda < 0.5$ earlier because for any such $\lambda$, selecting any candidate in the $\lambda$-weighted uncovered set can have distortion worse than $5$. Now we show lower bounds for all $\lambda \in [0, 1]$.
\begin{theorem}
For any $\lambda \in [0, 1]$, choosing the alphabetically smallest candidate in the $\lambda$-weighted uncovered set has distortion of at least $2 + \sqrt{5}$.
\label{theorem:lower}
\end{theorem}
\begin{proof}
The tight examples are illustrated in Figure~\ref{fig:tightness}. In both cases, all points lie on a straight line. In the left case, let $p$ portion of voters be of type $v_1$ and $1 - p$ portion of voters be of type $v_2$.\footnote{Technically this is only doable for $\lambda \in \mathbb{Q}$, but we can add an arbitrarily small portion of other voters to make $A$ in the $\lambda$-weighted uncovered set.} When $p \leq 1 - \varphi$, the distortion of choosing $A$ is at least $\frac{(1 - \varphi) \cdot 1 + \varphi \cdot 2}{(1 - \varphi) \cdot 1 + \varphi \cdot 0} = 2 + \sqrt{5}$. Since $A$ is in the $p$-weighted uncovered set and $(1 - p)$-weighted uncovered set, we rule out $\lambda \leq 1 - \varphi$ and $\lambda \geq \varphi$.

When $\lambda \in (1 - \varphi, \varphi)$, consider the right case in Figure~\ref{fig:tightness}. Let $1 - \lambda$ portion of voters be of type $v_1$ and $\lambda$ portion of voters be of type $v_2$. $A$ is in the $\lambda$-weighted uncovered set. However, the distortion for choosing $A$ comparing to $B$ is at least $\frac{(1 - \lambda) \cdot 2 + \lambda \cdot 3}{(1 - \lambda) \cdot 0 + \lambda \cdot 1} = 3 + 2 \cdot \frac{1 - \lambda}{\lambda}$. When $\lambda < \varphi$, the distortion is greater than $2 + \sqrt{5}$.
\end{proof}

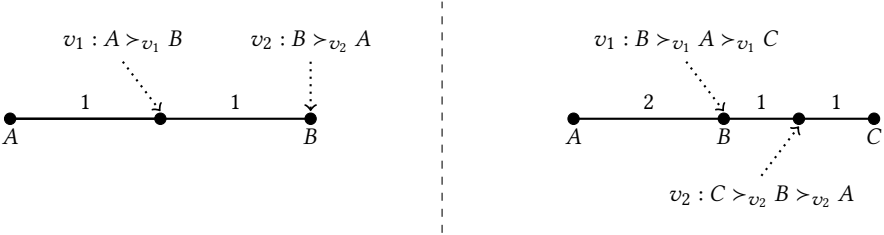
\begin{figure}[H]
\centering
\begin{tikzpicture}[scale=0.5]
\draw [fill=black] (0,0) circle [radius=0.15];
\draw [fill=black] (4,0) circle [radius=0.15];
\draw [fill=black] (8,0) circle [radius=0.15];

\draw [thick] (0, 0) to (4, 0);
\draw [thick] (0, 0) to (8, 0);

\draw [fill=black] (15,0) circle [radius=0.15];
\draw [fill=black] (19,0) circle [radius=0.15];
\draw [fill=black] (21,0) circle [radius=0.15];
\draw [fill=black] (23,0) circle [radius=0.15];

\draw [thick] (15, 0) to (19, 0);
\draw [thick] (19, 0) to (21, 0);
\draw [thick] (21, 0) to (23, 0);

\node [below] at (0,0) {\small $A$};
\node [below] at (8,0) {\small $B$};

\node [below] at (15,0) {\small $A$};
\node [below] at (19,0) {\small $B$};
\node [below] at (23,0) {\small $C$};

\draw [thick,dotted,->] (3, 1.5) to (4, 0.2);
\draw [thick,dotted,->] (8, 1.5) to (8, 0.2);
\draw [thick,dotted,->] (18, 1.5) to (19, 0.2);
\draw [thick,dotted,->] (20, -1.5) to (21, -0.2);

\node [above] at (2, 0) {\small $1$};
\node [above] at (6, 0) {\small $1$};
\node [above] at (17, 0) {\small $2$};
\node [above] at (20, 0) {\small $1$};
\node [above] at (22, 0) {\small $1$};

\node [above] at (3, 1.5) {\small $v_1 : A \succ_{v_1} B$};
\node [above] at (8, 1.5) {\small $v_2 : B \succ_{v_2} A$};
\node [above] at (18, 1.5) {\small $v_1 : B \succ_{v_1} A \succ_{v_1} C$};
\node [below] at (20, -1.5) {\small $v_2 : C \succ_{v_2} B \succ_{v_2} A$};

\draw [dashed] (11.5, -3) to (11.5, 3.08);
\end{tikzpicture}
\caption{Tight Examples for \textsc{WeightedUncovered} with Other Parameters}
\label{fig:tightness}
\end{figure}

The \emph{fairness ratio} of social choice rules was proposed in \cite{goel2017metric}:
\begin{definition}
The \emph{fairness ratio} $\Phi(f)$ of a social choice rule $f$ is:
\[
\Phi(f) = \sup_{\bm{\sigma}} \sup_{d \text{ consistent with } {\bm{\sigma}}} \max_{1 \leq k \leq m} \frac{\max_{R \subseteq \mathcal{V} : |R| = k} d(v, f(\bm{\sigma}))}{\min_{A \in \mathcal{C}} \max_{R \subseteq \mathcal{V} : |R| = k} d(v, A)}.
\]
\end{definition}
The fairness ratio $\Phi(f)$ is at least the distortion $\Delta(f)$ because it approximates the sum of $k$ largest costs for any $k$ including $m$.
The distortion of $\textsc{Copeland}$ is $5$~\cite{anshelevich2018approximating} and its fairness ratio is $5$~\cite{goel2017metric} as well. Since $\textsc{WeightedUncovered}$ improves the distortion, one may ask whether it improves the fairness ratio as well. We answer this in the negative:
\begin{theorem}
For any $\lambda \in [0, 1]$, choosing the alphabetically smallest candidate in the $\lambda$-weighted uncovered set has a fairness ratio of at least $5$.
\label{theorem:fairness_lower}
\end{theorem}
\begin{proof}
When $\lambda = 0$ or $\lambda = 1$, the uncovered set contains every candidate so the fairness ratio is $+\infty$. Otherwise, consider a setting where $\mathcal{C} = \{A, B, C\}$ and $\mathcal{V}$ contains multiple copies of $v_1$ and $v_2$. The distances between pairs of points are given in Table~\ref{tab:fairness}. One can verify they satisfy the triangle inequality.

The preference of $v_1$ is $C \succ B \succ A$ and that of $v_2$ is $C \succ A \succ B$. Let $\lambda m$ voters be of $v_1$ type and $(1 - \lambda) m$ voters be of $v_2$ type.\footnote{Again, this is only doable for $\lambda \in \mathbb{Q}$, but we can add an arbitrarily small portion of other voters to make the $\lambda$-weighted uncovered set include $A$. Instead of looking at the worst-off $1$ voter in the end, we look at the worst-off $\varepsilon$ portion of voters for a small $\varepsilon$.} $A$ is in the $\lambda$-weighted uncovered set, but the fairness ratio for choosing $A$ is at least $5$ comparing to $B$, given by the worst-off $1$ voter ($v_1$).
\end{proof}

\begin{table}[H]
    \begin{tabular}{| l | l | l | l | l | l |}
    \hline
    Points & $A$ & $B$ & $C$ & $v_1$ & $v_2$ \\ \hline
    $A$ & $0$ & $4$ & $4$ & $\textcolor{myred}{\bm{5}}$ & $\textcolor{myred}{\bm{3}}$ \\ \hline
    $B$ & $4$ & $0$ & $2$ & $\textcolor{myred}{\bm{1}}$ & $\textcolor{myred}{\bm{1}}$  \\ \hline
    $C$ & $4$ & $2$ & $0$ & $1$ & $3$ \\ \hline
    $v_1$ & $5$ & $1$ & $1$ & $0$ & $2$\\ \hline
    $v_2$ & $3$ & $1$ & $3$ & $2$ & $0$ \\ \hline
    \end{tabular}
\caption{Adjacency Matrix of the Tight Example for Fairness}
\label{tab:fairness}
\end{table}

\section{Matching-Uncovered-Set Voting Rule}
\label{sec:matching}
In this section, we present the \textsc{MatchingUncovered} rule, which is based on an alternative generalization of uncovered sets. We introduce \emph{matching uncovered sets} and the \textsc{MatchingUncovered} rule in Section~\ref{sec:matching_description}, after which we analyze its distortion (Theorem~\ref{theorem:uncovered_set_3}) in Section~\ref{sec:matching_analysis}. In Section~\ref{sec:conjecture}, we propose a conjecture (Conjecture~\ref{conjecture}) which implies the distortion of \textsc{MatchingUncovered} is $3$. We use computer-assisted search to verify it holds when at most $14$ entities  are involved (Theorem~\ref{theorem:cv_small}). Finally, we give a sufficient condition (Lemma~\ref{lem:reduction_2}) for the matching uncovered set to be nonempty in Section~\ref{sec:symmetric}. We use this new condition to show when the weighted tournament graph is cyclically symmetric, every social choice rule has distortion of at most $3$ (Theorem~\ref{theorem:symmetric}).

\subsection{Description}
\label{sec:matching_description}
Recall that as in Definition~\ref{def:PQ}, $P_v(A)$ is the set of candidates that $v$ likes as least as much as $A$ and $Q_v(A)$ is the set of candidates that $v$ likes as most as much as $A$.
\begin{definition}
The \emph{matching uncovered set} is the set of candidates $A \in \mathcal{C}$ such that: For any other candidate $B \neq A$, there is a perfect matching in the bipartite graph $G(A, B)$ constructed in the following way:
\begin{itemize}
\item The vertices on the left and vertices on the right are both $\mathcal{V}$.
\item For a vertex $v$ on the left and a vertex $v'$ on the right, we draw an edge between them if $P_v(B) \cap Q_{v'}(A) \neq \varnothing$.
\end{itemize}
\end{definition}

\begin{example}
Let $\mathcal{C} = \{A, B, C\}$ and $\mathcal{V} = \{1, 2, 3\}$. The preference profile is:
\begin{center}
$A \succ_1 B \succ_1 C$

$B \succ_2 C \succ_2 A$

$C \succ_3 A \succ_3 B$
\end{center}
$A$ is in the matching uncovered set because $G(A, B)$, $G(A, C)$ both have perfect matchings, as shown in Figure~\ref{fig:matching_uncovered_set_1} and Figure~\ref{fig:matching_uncovered_set_2}. By cyclic symmetry, $B$ and $C$ are in the matching uncovered set too.
\begin{figure}[H]
\centering
\begin{tikzpicture}[scale=0.3]
\node [left] at (0,8) {$P_1(B) = \{A, B\}$};
\node [left] at (0,4) {$P_2(B) = \{B\}$};
\node [left] at (0,0) {$P_3(B) = \{A, B, C\}$};
\node [right] at (7,8) {$Q_1(A) = \{A, B, C\}$};
\node [right] at (7,4) {$Q_2(A) = \{A\}$};
\node [right] at (7,0) {$Q_3(A) = \{A, B\}$};

\draw (1, 8) to (6, 8);
\draw [ultra thick, myred] (1, 8) to (6, 4);
\draw (1, 8) to (6, 0);
\draw [ultra thick, myred] (1, 4) to (6, 8);
\draw (1, 4) to (6, 0);
\draw (1, 0) to (6, 8);
\draw (1, 0) to (6, 4);
\draw [ultra thick, myred] (1, 0) to (6, 0);

\draw [fill=white](1,8) circle [radius=1];
\draw [fill=white](1,4) circle [radius=1];
\draw [fill=white](1,0) circle [radius=1];

\draw [fill=white](6,8) circle [radius=1];
\draw [fill=white](6,4) circle [radius=1];
\draw [fill=white](6,0) circle [radius=1];

\node at (1,8) {$1$};
\node at (1,4) {$2$};
\node at (1,0) {$3$};

\node at (6,8) {$1$};
\node at (6,4) {$2$};
\node at (6,0) {$3$};
\end{tikzpicture}
\caption{The Bipartite Graph $G(A, B)$}
\label{fig:matching_uncovered_set_1}
\end{figure}
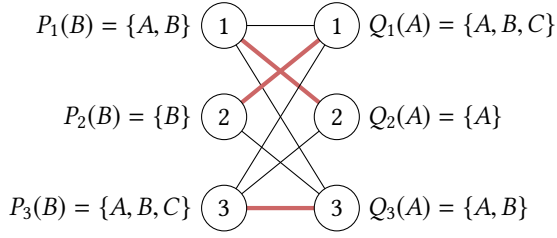

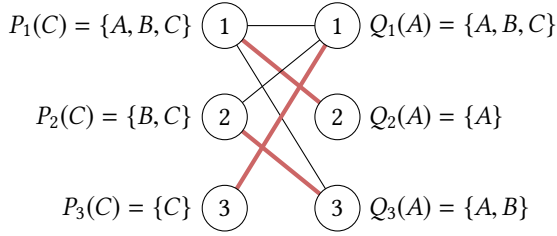
\begin{figure}[H]
\centering
\begin{tikzpicture}[scale=0.3]
\node [left] at (0,8) {$P_1(C) = \{A, B, C\}$};
\node [left] at (0,4) {$P_2(C) = \{B, C\}$};
\node [left] at (0,0) {$P_3(C) = \{C\}$};
\node [right] at (7,8) {$Q_1(A) = \{A, B, C\}$};
\node [right] at (7,4) {$Q_2(A) = \{A\}$};
\node [right] at (7,0) {$Q_3(A) = \{A, B\}$};

\draw (1, 8) to (6, 8);
\draw [ultra thick, myred] (1, 8) to (6, 4);
\draw (1, 8) to (6, 0);
\draw (1, 4) to (6, 8);
\draw [ultra thick, myred] (1, 4) to (6, 0);
\draw [ultra thick, myred] (1, 0) to (6, 8);

\draw [fill=white](1,8) circle [radius=1];
\draw [fill=white](1,4) circle [radius=1];
\draw [fill=white](1,0) circle [radius=1];

\draw [fill=white](6,8) circle [radius=1];
\draw [fill=white](6,4) circle [radius=1];
\draw [fill=white](6,0) circle [radius=1];

\node at (1,8) {$1$};
\node at (1,4) {$2$};
\node at (1,0) {$3$};

\node at (6,8) {$1$};
\node at (6,4) {$2$};
\node at (6,0) {$3$};
\end{tikzpicture}
\caption{The Bipartite Graph $G(A, C)$}
\label{fig:matching_uncovered_set_2}
\end{figure}
\end{example}

\begin{definition}
\textsc{MatchingUncovered} is a social choice rule that picks the alphabetically smallest candidate in the matching uncovered set. It picks the alphabetically smallest candidate if the matching uncovered set is empty.
\end{definition}

\subsection{Distortion Analysis}
\label{sec:matching_analysis}
The next theorem is the main reason why we introduce matching uncovered sets.
\begin{theorem}
Selecting any candidate from the matching uncovered set guarantees distortion of at most $3$.
\label{theorem:uncovered_set_3}
\end{theorem}

\begin{corollary}
If the matching uncovered set is nonempty, \textsc{MatchingUncovered} has distortion of at most $3$.
\label{cor:matching_uncovered_3}
\end{corollary}

Before proving Theorem~\ref{theorem:uncovered_set_3}, we introduce a new notation specifically for proving distortion of $3$. Fixing the preference orderings $\{\sigma_k\}_{k = 1}^m$, define
$$U(v, A, B, d) := d(A, v) - 3d(B, v).$$
\begin{lemma}
A social choice rule $f$ has distortion of at most 3 if and only if for any voting profile $\bm{\sigma} = \{\sigma_k\}_{k = 1}^m$ and for any other candidate $B \neq f\left(\bm{\sigma}\right)$,
$$\sum_{v \in \mathcal{V}} U(v, f\left(\bm{\sigma}\right), B, d) \leq 0$$
for any metric $d$ that is consistent with $\bm{\sigma}$.
\label{lem:trivial_reduction}
\end{lemma}
\begin{proof}
Let $A = f(\bm{\sigma})$. A social choice rule $f$ has distortion of at most $3$ if and only if for any profile $\bm{\sigma}$ and any metric $d$ consistent with $\bm{\sigma}$, $\sum_{v \in \mathcal{V}} d(A, v) \leq 3\sum_{v \in \mathcal{V}} d(B, v)$ for any other $B \in \mathcal{C}$, \emph{i.e.}, $\sum_{v \in \mathcal{V}} U(v, A, B, d) \leq 0$.
\end{proof}

\begin{lemma}
For a voter $v$ and any consistent metric $d$:
\begin{enumerate}
\item If $v$ has preference $Y \succ_v B \succ_v A \succ_v X$, then $U(v, A, B, d) \leq d(B, X) - d(B, Y)$.
\item If $v$ has preference $Y \succ_v B \succ_v A$, then $U(v, A, B, d) \leq d(A, B) - d(B, Y)$.
\item If $v$ has preference $B \succ_v A \succ_v X$, then $U(v, A, B, d) \leq d(B, X)$.
\item If $v$ has preference $A \succ_v B$, then $U(v, A, B, d) \leq -d(A, B)$.
\item For any $v$, $U(v, A, B, d) \leq d(A, B)$.
\end{enumerate}
\label{lem:inequalities}
\end{lemma}
\begin{proof}
By the preference of $v$ and triangle inequality:
\begin{enumerate}
\item If $v$ has preference $Y \succ_v B \succ_v A \succ_v X$, then
\begin{align*}
\quad U(v, A, B, d) &= (d(v, A) - d(v, B)) - (d(v, B) + d(v, B))\\
&\leq (d(v, X) - d(v, B)) - (d(v, B) + d(v, Y))\\
&\leq d(B, X) - d(B, Y).
\end{align*}
\item If $v$ has preference $Y \succ_v B \succ_v A$, then
\begin{align*}
U(v, A, B, d) &= (d(v, A) - d(v, B)) - (d(v, B) + d(v, B))\\
&\leq (d(v, A) - d(v, B)) - (d(v, B) + d(v, Y))\\
&\leq d(A, B) - d(B, Y).
\end{align*}
\item If $v$ has preference $B \succ_v A \succ_v X$, then
\begin{align*}
U(v, A, B, d) &= (d(v, A) - d(v, B)) - 2d(v, B)\\
&\leq d(v, X) - d(v, B)\\
&\leq d(B, X).
\end{align*}
\item If $v$ has preference $A \succ_v B$, then
\begin{align*}
U(v, A, B, d) &= 2(d(v, A) - d(v, B)) - (d(v, A) + d(v, B))\\
&\leq -(d(v, A) + d(v, B))\\
&\leq -d(A, B).
\end{align*}
\item For any $v$,
\begin{align*}
U(v, A, B, d) &= d(v, A) - d(v, B) - 2d(v, B)\\
&\leq d(v, A) - d(v, B)\\
&\leq d(A, B).\qedhere
\end{align*}
\end{enumerate}
\end{proof}

\begin{proof}[Proof of Theorem~\ref{theorem:uncovered_set_3}]
Because of Lemma~\ref{lem:trivial_reduction}, we only need to show $\sum_{v \in \mathcal{V}} U(v, A, B, d) \leq 0$ if there is a perfect matching in $G(A, B)$.

We divide the set of voters into two disjoint subsets: $AB$ (those who prefer $A$ to $B$) and $BA$. By (4) in Lemma~\ref{lem:inequalities}, we have
\begin{align*}
\sum_{v \in \mathcal{V}} U(v, A, B, d) &= \sum_{v \in AB} U(v, A, B, d) + \sum_{v \in BA} U(v, A, B, d)\\
&\leq -|AB| \cdot d(A, B) + \sum_{v \in BA} U(v, A, B, d).
\end{align*}
For a voter $v$, we use $L_v$ to denote the vertex on the left representing $P_v(B)$ and $R_v$ the vertex on the right representing $Q_v(A)$. Because there is a perfect matching in the graph $G(A, B)$, there has to be a matching of size at least $m - 2|AB|$ within the vertices representing $BA$: $\{L_v \mid v \in BA\} \cup \{R_v \mid v \in BA\}$. We write $M(u)$ to denote the boolean variable whether the vertex $u$ is matched in this smaller matching. If $M(u)$ is true, we write $K(u)$ to denote the common candidate that $u$ uses in the matching. Note that the common candidate cannot be $A$ or $B$ between vertices representing $BA$. We have
\begin{align*}
&\sum_{v \in BA} U(v, A, B, d)\\
= &\sum_{v \in BA: M(L_v) \land M(R_v)} U(v, A, B, d) + \sum_{v \in BA: M(L_v) \land \lnot M(R_v)} U(v, A, B, d)\\
&+ \sum_{v \in BA: \lnot M(L_v) \land M(R_v)} U(v, A, B, d) + \sum_{v \in BA: \lnot M(L_v) \land \lnot M(R_v)} U(v, A, B, d)\\
\leq& \sum_{v \in BA: M(L_v) \land M(R_v)} (d(B, K(R_v)) - d(B, K(L_v))) + \sum_{v \in BA: M(L_v) \land \lnot M(R_v)} (d(A, B) - d(B, K(L_v)))\\
&+ \sum_{v \in BA: \lnot M(L_v) \land M(R_v)} d(B, K(R_v)) + \sum_{v \in BA: \lnot M(L_v) \land \lnot M(R_v)} d(A, B)\\
=& \sum_{v \in BA: M(L_v) \land \lnot M(R_v)} d(A, B) + \sum_{v \in BA: \lnot M(L_v) \land \lnot M(R_v)} d(A, B)\\
=& \sum_{v \in BA: \lnot M(R_v)} d(A, B).
\end{align*}
The inequality above is by Lemma~\ref{lem:inequalities}. The second last step is because the terms of $d(B, K(R_v))$ and $-d(B, K(L_v))$ cancel out in a matching. Therefore,
\begin{align*}
\sum_{v \in \mathcal{V}} U(v, A, B, d) &\leq -|AB| \cdot d(A, B) + \sum_{v \in BA} U(v, A, B, d)\\
&\leq -|AB| \cdot d(A, B) + |\{v \in BA \mid \lnot M(R_v)\}| \cdot d(A, B)\\
&\leq 0.
\end{align*}
The last step is because $|\{v \in BA \mid \lnot M(R_v)\}|$ is the number of unmatched vertices on the right in the smaller matching within $BA$, which is at most $|BA| - (m - 2|AB|) = |AB|$.
\end{proof}

\subsection{Existence Conjecture}
\label{sec:conjecture}
Now we present a combinatorial conjecture which, if true, will guarantee $\textsc{MatchingUncovered}$ has distortion of $3$.
\begin{conjecture}
For any $\mathcal{C} = \{X_1, X_2, \ldots, X_n\}$ and $\mathcal{V}$, for any voting profile, at least one of the bipartite graphs $G(X_1, X_2), G(X_2, X_3), G(X_3, X_4), \ldots, G(X_{n-1}, X_n)$ and $G(X_n, X_1)$ has a perfect matching.
\label{conjecture}
\end{conjecture}

Conjecture~\ref{conjecture} immediately implies that \textsc{MatchingUncovered} has distortion of 3: If the matching uncovered set is empty, we can start at any vertex $A$, find the graph $G(A, B)$ without a perfect matching, then go to $B$ and repeat until we find a cycle $X_1, X_2, \ldots, X_t$. We delete any candidate not on the cycle from the voting profile and maintain relative orders for other candidates. Since $G(X_k, X_{k+1})$ does not have a perfect matching before the deletions, $G(X_k, X_{k+1})$ constructed after the deletions does not have a perfect matching either. However, Conjecture~\ref{conjecture} denies the existence of such a cycle. ($\mathcal{C}$ in Conjecture~\ref{conjecture} is the set of candidates on the cycle.)

Though we are not able to prove or disprove Conjecture~\ref{conjecture}, we can show it when the number of candidates or the number of voters is at most $3$.
\begin{theorem}
When $|\mathcal{C}| \leq 3$, Conjecture~\ref{conjecture} holds.
\label{lem:c_leq_3}
\end{theorem}
\begin{proof}
When $|\mathcal{C}| = 2$, a candidate $A$ is preferred to the other candidate $B$ by at least half of the voters. There is a perfect matching in $G(A, B)$. When $|\mathcal{C}| = 3$, assume a counterexample exists. Let $\mathcal{C} = \{A = X_1, B = X_2, C = X_3\}$. Let $a = \frac{|ABC|}{m}, b = \frac{|BCA|}{m}, c =\frac{|CAB|}{m}, x = \frac{|ACB|}{m}, y = \frac{|BAC|}{m}, z = \frac{|CBA|}{m}$. In order that none of $G(A, B), G(B, C), G(C, A)$ has a perfect matching, we need

\[
  \begin{cases}
    2(a + c + x) + \min(y, z) < 1\\
    2(b + a + y) + \min(z, x) < 1\\
    2(c + b + z) + \min(x, y) < 1
  \end{cases}.
\]
Notice that if we move the weight for $a$ ($b$, $c$ respectively) to $z$ ($x$, $y$ respectively), all three expressions on the left become weakly smaller. Thus we can assume $a = b = c = 0$ without loss of generality. Therefore we need
\[
  \begin{cases}
    2x + \min(y, z) < 1\\
    2y + \min(z, x) < 1\\
    2z + \min(x, y) < 1
  \end{cases}
\]
while $x + y + z = 1$. Without loss of generality, assume $x = \max(x, y, z)$. Then $2x + \min(y, z) = 2\max(x, y, z) + \min(x, y, z) \geq x + y + z = 1$. Thus the set of inequalities cannot be simultaneously satisfied.
\end{proof}

\begin{theorem}
When $|\mathcal{V}| \leq 3$, Conjecture~\ref{conjecture} holds.
\label{lem:v_leq_3}
\end{theorem}
\begin{proof}
When $|\mathcal{V}| \leq 2$, the first voter must prefer some $X_k$ to $X_{k+1}$ (or $X_n$ to $X_1$). There is a perfect matching in $G(X_k, X_{k+1})$ then. When $|\mathcal{V}| = 3$, we prove by contradiction and assume such a counterexample exists. Let $\mathcal{V} = \{v_1, v_2, v_3\}$. For each $i = 1, 2, \ldots, n$, at least $2$ voters prefer $X_{i + 1}$ to $X_i$. If all voters prefer $X_{i + 1}$ to $X_i$, we can remove $X_i$ in the cycle and there will not be a perfect matching in $G(X_{i-1}, X_{i+1})$. Therefore without loss of generality, for each $i = 1, 2, \ldots, n$, we assume exactly $2$ voters prefer $X_{i + 1}$ to $X_i$. Again without loss of generality, let $X_1$ be the top choice of $v_3$. Then $v_1$ and $v_2$ must prefer $X_2$ to $X_1$. Let $k$ be the smallest one such that $v_3$ prefers $X_{k+1}$ to $X_k$. Among candidates $X_1, X_2, \ldots, X_k$, the preference of $v_1$ and $v_2$ is $X_k \succ X_{k-1} \succ \cdots \succ X_1$ while the preference of $v_3$ is $X_1 \succ X_2 \succ \cdots \succ X_k$. This $k$ must be strictly less than $n$ for $G(X_n, X_1)$ not to have a perfect matching. Notice $G(X_k, X_{k+1})$ must have a perfect matching: there is a voter preferring $X_{k}$ to $X_{k+1}$ contributing $2$ matches. The other two voters contribute one match because of the common $X_1$ ($v_1$ or $v_2$ has $X_k \succ X_1$ and $v_3$ has $X_1 \succ X_{k+1}$).
\end{proof}

Besides the cases where $\min(|\mathcal{C}|, |\mathcal{V}|) \leq 3$, we performed computer-assisted search to enumerate all possible voting profiles for small $(|\mathcal{C}|, |\mathcal{V}|)$ pairs, and we were not able to find any counterexample. Combining with Lemma~\ref{lem:c_leq_3} and Lemma~\ref{lem:v_leq_3}, we show the following result:
\begin{theorem}
When $(|\mathcal{C}|, |\mathcal{V}|) \in \mathcal{E}$, Conjecture~\ref{conjecture} holds. Here
\begin{align*}
\mathcal{E} = &([3] \times \mathbb{N}) \cup ([4] \times [20]) \cup ([5] \times [11]) \cup ([6] \times [8]) \cup ([7] \times [7])\\
&\cup ([8] \times [6]) \cup ([9] \times [5]) \cup ([10] \times [4]) \cup (\mathbb{N} \times [3]).
\end{align*}
\label{theorem:cv_small}
\end{theorem}
In particular, when $|\mathcal{C}| + |\mathcal{V}| \leq 14$, Conjecture~\ref{conjecture} holds. In other words, if the circumference (size of the largest cycle) of the tournament graph plus the number of voters is at most $14$, \textsc{MatchingUncovered} achieves distortion of $3$.

\subsection{Cyclically Symmetric Weighted Tournament Graphs}
\label{sec:symmetric}
In this section, we prove that for any instance inducing a cyclically symmetric weighted tournament graph, the distortion of selecting any candidate is at most $3$. Recall that by cyclic symmetry, we mean there exists a cyclic permutation $\tau$ on the vertices, so that any edge $(u, v)$ has the same weight as $(\tau(u), \tau(v))$.
\begin{theorem}
For any preference profile that induces a cyclically symmetric weighted tournament graph, every social choice rule achieves distortion of at most $3$.
\label{theorem:symmetric}
\end{theorem}
Our proof crucially relies on the notion of matching uncovered sets. We first introduce a simpler test for the existence of perfect matchings in $G(A, B)$.
\begin{lemma}
Let $w(X, Y) = \frac{|XY|}{m}$ be the weight of edge $XY$ in the weighted tournament graph. If there is no perfect matching in $G(A, B)$, then in the tournament graph,
\[
(w(A, B), w(B, A)) - \left(\bigcup_{C \in (\mathcal{C} - \{A, B\})} [w(C, A), w(C, B)]\right)
\]
must be nonempty. Here $[a, b]$ denotes the closed interval $\{x \in \mathbb{R} \mid a \leq x \leq b\}$ and $(a, b)$ denotes the open interval.
\label{lem:reduction_2}
\end{lemma}
\begin{proof}
If there is no perfect matching in $G(A, B)$, there must be an independent set of at least $m + 1$ vertices. Therefore, for sufficiently small $\varepsilon > 0$, there exists $x \in (0, 1)$, so that at least $x + \varepsilon$ portion of vertices on the right are disjoint with at least $1 - x + \varepsilon$ portion of vertices on the left. As a result, we can select $Q_v(A)$'s for at least $x + \varepsilon$ portion of $v \in \mathcal{V}$ and $P_{v'}(B)$'s for at least $1 - x + \varepsilon$ portion of $v' \in \mathcal{V}$, and make each $Q_v(A)$ disjoint with each $P_{v'}(B)$.

This $x$ must be in the interval $(w(A, B), w(B, A))$, otherwise the selected $Q_v(A)$'s and $P_{v'}(B)$'s will share a common $A$ or $B$. For any other candidate $C$, either $C$ does not appear in the selected $Q_v(A)$'s, or it does not appear in the selected $P_{v'}(B)$'s. It implies either $w(A, C) \leq 1 - x - \varepsilon$, or $w(C, B) \leq x - \varepsilon$. (Otherwise, there are too many $Q_v(A)$'s or $P_{v'}(B)$'s containing $C$, which makes the selection impossible.) Thus, $x < W(C, A)$ or $x > W(C, B)$.
\end{proof}

\begin{lemma}
For any preference profile that induces a cyclically symmetric weighted tournament graph, every candidate is in the matching uncovered set.
\label{lem:symmetric_uncovered}
\end{lemma}
\begin{proof}
We only need to show every $G(A, B)$ has a perfect matching. We will use Lemma~\ref{lem:reduction_2} and show
\[
(w(A, B), w(B, A)) - \left(\bigcup_{C \in (\mathcal{C} - \{A, B\})} [w(C, A), w(C, B)]\right)
\]
is empty. Let $\tau$ be the cyclic permutation that shows the cyclically symmetric structure of the weighted tournament graph. Let $B = \tau^k(A)$ and $g$ be the greatest common divisor of $k$ and $n$. Let $h = n/g$ be the minimum positive integer that $\tau^{kh}$ is the identity permutation. If $h = 2$, then $w(A, B) = w(B, A)$ so $(w(A, B), w(B, A)) = \varnothing$. Otherwise,
\begin{align*}
\left(\bigcup_{C \in (\mathcal{C} - \{A, B\})} [w(C, A), w(C, B)]\right) &\supseteq \bigcup_{j = 2}^{h-1} \left[w\left(\tau^{jk}(A), A\right), w\left(\tau^{jk}(A), \tau^{k}(A)\right)\right]\\
&= \bigcup_{j = 2}^{h-1} \left[w\left(\tau^{jk}(A), A\right), w\left(\tau^{(j - 1)k}(A), A\right)\right]\\
&\supseteq \left[w\left(\tau^{(h - 1)k}(A), A\right), w\left(\tau^{(2 - 1)k}(A), A\right)\right]\\
&= \left[w\left(\tau^{hk}(A), \tau^{k}(A)\right), w\left(\tau^{k}(A), A\right)\right]\\
&= \left[w\left(A, B\right), w\left(B, A\right)\right].
\end{align*}
Thus $(w(A, B), w(B, A)) - \left(\bigcup_{C \in (\mathcal{C} - \{A, B\})} [w(C, A), w(C, B)]\right) = \varnothing$ so $G(A, B)$ has a perfect matching by Lemma~\ref{lem:reduction_2}.
\end{proof}

\begin{proof}[Proof of Theorem~\ref{theorem:symmetric}]
By Lemma~\ref{lem:symmetric_uncovered} and Theorem~\ref{theorem:uncovered_set_3}, if the weighted tournament graph is symmetric, every candidate is in the matching uncovered set, so selecting any candidate gives distortion of $3$.
\end{proof}

In \cite{anshelevich2018approximating}, the authors show the class of (unweighted) tournament rules has a distortion lower bound of $5$. They construct an instance with a cyclically symmetric (unweighted) tournament graph where selecting one of the candidates has distortion of $5$. The class of tournament rules cannot always avoid selecting the bad candidate so a lower bound of $5$ is established. Theorem~\ref{theorem:symmetric} shows that the exact same technique with cyclically symmetric tournament graphs cannot be applied to establish a non-trivial lower bound for the class of weighted tournament rules. It remains an open question that whether weighted tournament rules are strictly weaker than $\textsc{OptimalLP}$ in terms of distortion.

We note that, in Conjecture~\ref{conjecture}, we cannot replace the condition that $G(X_k, X_{k+1})$ has no perfect matching by that in Lemma~\ref{lem:reduction_2}. Otherwise the conjecture is false by the following example.
\begin{example}
Consider the weighted tournament graph shown in Figure~\ref{fig:counterexample_relaxation_2}.
\begin{figure}[H]
\centering
\begin{tikzpicture}[scale=0.3]
\draw [thick,->] (0, 0) to (0,10*0.87);
\draw [thick,->] (0, 10) to (10*0.87,10);
\draw [thick,->] (10, 10) to (10,1.3);
\draw [thick,->] (10, 0) to (1.3,0);

\draw [thick,->] (0,0) to (10*0.905,10*0.905);
\draw [thick,->] (0,10) to (10*0.905,0.95);

\draw [fill=white](0,0) circle [radius=1.3];
\draw [fill=white](0,10) circle [radius=1.3];
\draw [fill=white](10,10) circle [radius=1.3];
\draw [fill=white](10,0) circle [radius=1.3];

\node at (0,0) {\small $A$};
\node at (0,10) {\small $B$};
\node at (10,10) {\small $C$};
\node at (10,0) {\small $D$};

\node at (-1,5) {\small $0.3$};
\node at (5,10.8) {\small $0.35$};
\node at (11.2,5) {\small $0.35$};
\node at (5,-0.8) {\small $0.3$};

\node at (2.5,6.2) {\small $0.45$};
\node at (7.5,6.2) {\small $0.5$};
\end{tikzpicture}
\caption{A Counterexample for an Oversimplification with Lemma~\ref{lem:reduction_2}}
\label{fig:counterexample_relaxation_2}
\end{figure}
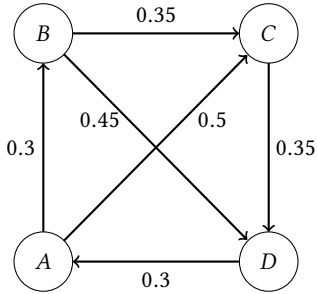
Only one edge between each pair of vertices is drawn for visibility, and the weight of the other one is simply $1$ minus that of the drawn one. The graph is realizable by the following voting profile:
\begin{eqnarray*}
&i = 1, 2, \ldots, 35: &B \succ_i A \succ_i D \succ_i C\\
&i = 36, 37, \ldots, 45: &C \succ_i B \succ_i A \succ_i D\\
&i = 46, 47, \ldots, 60: &D \succ_i C \succ_i B \succ_i A\\
&i = 61, 62, \ldots, 70: &C \succ_i D \succ_i B \succ_i A\\
&i = 71, 72, \ldots, 85: &A \succ_i D \succ_i C \succ_i B\\
&i = 86, 87, \ldots, 95: &C \succ_i A \succ_i D \succ_i B\\
&i = 96, 97, \ldots, 100: &C \succ_i D \succ_i A \succ_i B
\end{eqnarray*}
There are perfect matchings in $G(C, D)$ and $G(D, A)$, but
\begin{align*}
(w(A, B)&, w(B, A)) - [w(C, A), w(C, B)] - [w(D, A), w(D, B)]\\ &= (0.3, 0.7) - [0.5, 0.65] - [0.3, 0.55] \neq \varnothing,\\
(w(B, C)&, w(C, B)) - [w(D, B), w(D, C)] - [w(A, B), w(A, C)]\\ &= (0.35, 0.65) - [0.55, 0.65] - [0.3, 0.5] \neq \varnothing,\\
(w(C, D)&, w(D, C)) - [w(A, C), w(A, D)] - [w(B, C), w(B, D)]\\ &= (0.35, 0.65) - [0.5, 0.7] - [0.35, 0.45] \neq \varnothing,\\
(w(D, A)&, w(A, D)) - [w(B, D), w(B, A)] - [w(C, D), w(C, A)]\\ &= (0.3, 0.7) - [0.45, 0.7] - [0.35, 0.5] \neq \varnothing.
\end{align*}
\label{ex:counterexample_relaxation_2}
\end{example}

\section{Conclusions and Future Directions}
In this paper, we propose a deterministic social choice rule $\textsc{WeightedUncovered}$ which improves the best distortion upper bound from $5$ to $2 + \sqrt{5}$. We also show that this class of rules cannot improve the distortion beyond $2 + \sqrt{5}$ or improve the fairness ratio of \textsc{Copeland}. \textsc{MatchingUncovered}, on the other hand, may achieve better distortion as well as fairness ratio. The open question remains:
\begin{center}
What is the optimal distortion for the deterministic social choice problem?
\end{center}
We know $\textsc{OptimalLP}$ is a polynomial-time mechanism achieving optimal distortion, but rules of simpler forms may be more desirable in practice:
\begin{center}
Does $\textsc{OptimalLP}$ have strictly better distortion than $\textsc{WeightedUncovered}$, $\textsc{MatchingUncovered}$, and the class of weighted tournament rules?
\end{center}

\begin{acks}
Kamesh Munagala is supported by NSF grants CCF-1408784, CCF-1637397, and IIS-1447554; and research awards from Adobe and Facebook. Kangning Wang is supported by NSF grants IIS-1447554 and CCF-1637397.
\end{acks}

\bibliographystyle{ACM-Reference-Format}
\bibliography{ref}

\end{document}